%% file: main.tex
\documentclass{article}
\usepackage{graphicx} 

\usepackage[T1]{fontenc}
\usepackage[latin9]{inputenc}
\usepackage{amsmath}
\usepackage{amssymb}
\usepackage{mathtools}
\usepackage{amsthm}

\usepackage{appendix}
\usepackage{hyperref}
\hypersetup{
  colorlinks=true,
  linkcolor=blue,
  citecolor=blue,
  filecolor=blue,
  urlcolor=blue
}

\usepackage{algorithm}

\usepackage{algpseudocode}
\algrenewcommand\algorithmicrequire{\textbf{Parameters:}}

\usepackage{subcaption}

\newtheorem{lemma}{Lemma}[section]

\newtheorem{theorem}{Theorem}[section]
\newtheorem{corollary}{Corollary}[section]
\newtheorem{definition}{Definition}[section]
\newtheorem{remark}{Remark}[section]
\newtheorem{observation}{Observation}[section]

\newcommand{\indic}[1]{\mathbbm{1}\left[#1\right]}

\usepackage[numbers]{natbib}

\usepackage{fullpage}
\usepackage{graphicx}
\graphicspath{{figures/}}
\usepackage{bbm}
\usepackage[]{color-edits} 
\usepackage{thm-restate}

\newcommand{\reals}{\mathbb{R}}
\newcommand{\ALG}{\mathrm{ALG}}
\newcommand{\E}{\mathbb{E}}
\newcommand{\opt}{\text{OPT}}
\usepackage{tikz}

\usepackage{array}
\usepackage{multirow}

\usepackage[nice]{nicefrac}

\usepackage{orcidlink}
\usepackage{thm-restate}

\addauthor{TE}{magenta}

\addauthor{TG}{red}

\newcommand{\permutations}[1]{\mathcal{S}_{#1}}
\newcommand{\naturals}{\mathbbm{N}}

\newcommand{\class}{\mathcal{C}}

\title{The Competition Complexity of Prophet Secretary\footnote{T.~Ezra is supported by the Harvard University Center of Mathematical Sciences and Applications. T.~Garbuz is funded by the European Research Council (ERC) under the European Union's Horizon 2020 research and innovation program (grant agreement No. 101077862).}}

\author{Tomer Ezra\thanks{Harvard University, Cambridge, MA, USA. Email: \texttt{tomer@cmsa.fas.harvard.edu}. \orcidlink{0000-0003-0626-4851}} 
\and Tamar Garbuz\thanks{Tel Aviv University, Tel Aviv, Israel. Email: \texttt{garbuz@mail.tau.ac.il}. \orcidlink{0009-0009-3922-4160}}}

\date{}

\begin{document}

\maketitle
\begin{abstract}
We study the classic single-choice prophet secretary problem through a resource augmentation lens. Our goal is to bound the $(1-\epsilon)$-competition complexity for different classes of online algorithms. This metric asks for the smallest $k$ such that the expected value of the online algorithm on $k$ copies of the original instance, is at least a $(1 - \epsilon)$-approximation to the expected offline optimum on the original instance (without added copies).

We consider four natural classes of online algorithms: single-threshold, time-based threshold, activation-based, and general algorithms. 
We show that for single-threshold algorithms the $(1-\epsilon)$-competition complexity is $\Theta(\ln(\nicefrac{1}{\epsilon}))$ (as in the i.i.d. case).
Additionally, we demonstrate that time-based threshold and activation-based algorithms (which cover 
all previous approaches for obtaining competitive-ratios for the classic prophet secretary problem) yield a sub-optimal $(1-\epsilon)$-competition complexity of $\Theta\left(\frac{\ln(\nicefrac{1}{\epsilon})}{\ln\ln(\nicefrac{1}{\epsilon})}\right)$, which is strictly better than the class of single-threshold algorithms.
Finally, we find that the $(1-\epsilon)$-competition complexity of general adaptive algorithms is $\Theta(\sqrt{\ln(\nicefrac{1}{\epsilon})})$, which is in sharp contrast to $\Theta(\ln\ln(\nicefrac{1}{\epsilon}))$ in the i.i.d. case.

\end{abstract}

\pagenumbering{arabic}
\input{new-intro}

\input{new-related}

\input{model}

\input{secretary}

\input{single-threshold}

\input{blind-algorithms}

\input{non-adapt-hard}

\input{adaptive}
\input{example}

\bibliography{bib}
\bibliographystyle{abbrvnat}
\appendix

\input{app-case1}
\input{app-prelim}

\input{App-new}

\end{document}

%% file: new-intro.tex
\section{Introduction}
One of the most fundamental research directions in mechanism design is the tradeoff between simplicity and optimality of mechanisms. A canonical example of this tradeoff is the measure of \textit{competition complexity} in auctions, originated by \citet{BulowKlemperer96} who asked the question of how many additional bidders are needed to achieve using a simple mechanism the same revenue as the optimal (complicated) auction. In particular, \citet{BulowKlemperer96} showed that when selling a single item to a set of bidders with values drawn i.i.d. from some regular\footnote{A distribution with density $f$ and cumulative density $F$ is regular if its virtual value function $\phi(v):=v-\frac{1-F(v)}{f(v)}$ is non-decreasing.}  distribution, the (simple) second price auction with an additional bidder achieves at least the same expected revenue as the optimal distribution-tailored mechanism (the Myerson auction \citep{myerson1981optimal}). This result asserts that recruiting one more bidder is more important (in terms of revenue) than learning the exact distribution of the values of the bidders and tailoring the auction according to this information. The work of \citet{BulowKlemperer96} inspired a line of work \citep{eden2017competition,feldman201899,liu2018competition,beyhaghi2019optimal,cai202199,fu2019vickrey,babaioff2020bulow,brustle2023competition,derakhshan2024settling,DBLP:journals/corr/abs-2402-11084} that studied the competition complexity measure of more evolved scenarios such as auctions to bidders with independently but not identically distributed values, as well as auctions beyond single-item settings.
In particular, \citet{feldman201899} introduced the combination of competition complexity with approximation algorithms, where instead of surpassing the revenue of the optimal auction, they aim to use as few as possible additional bidders in a simple auction to obtain $(1-\epsilon)$-approximation of the expected revenue of the optimal mechanism.

Most of this line of work focused on two aspects of simple auctions. Namely, auctions that have less information on the priors (such as the result of \citet{BulowKlemperer96}) or auctions that are easy to implement computationally (such as selling items separately, or the grand bundle in a second price auction).
However, a key complexity aspect of implementing  mechanisms is the requirement to coordinate the bidders to participate simultaneously.

In this paper, we study the competition complexity measure in online auctions where bidders arrive sequentially and the seller needs to sell the item in an immediate and irrevocable fashion.
Sequential auctions are well-studied\footnote{Sequential auctions are also studied under the prophet inequality framework. The connection to prophet inequality was observed in \citet{chawla2010multi}, where they showed that implementing a prophet inequality algorithm on the virtual values leads to an approximation result against the expected revenue of the optimal Myerson's auction.} and are widely used in practice.
Moreover, sequential auctions are likely to attract more participants as there is no need for coordination between the bidders, and each bidder faces a simple decision of whether to purchase the item at a take-it-or-leave-it price.
\citet{brustle2023competition} introduced the question of the competition complexity of dynamic pricing with i.i.d. bidders, and showed that in contrast to  \citep{BulowKlemperer96}, the exact competition complexity of online auctions is unbounded, and analyzed the asymptotic behavior of the additional amount of bidders needed to achieve $(1-\epsilon)$-approximation of the optimal offline auction.
\citet{DBLP:journals/corr/abs-2402-11084} extended the result of \citep{brustle2023competition} to the case of independent but not identical bidders that arrive according to the block arrival model (where they ask how many sequential copies of the original instance are needed to approximate the expectation of the optimal auction on a single instance).
However, the assumption that participants arrive according to the block arrival model is very restrictive and captures only a few real-life scenarios.

The research question of our paper is to analyze the competition complexity of independent but not identical bidders that arrive in the more natural arrival model of a random order.
Sequential auctions under random arrival order relate heavily to the prophet secretary framework introduced by \citet{esfandiari2017prophet}, in which a decision-maker faces a series of random rewards $V_1,\ldots,V_n$ that arrive in a random order, and needs to select one of them in an immediate and irrevocable fashion. 
The rewards are drawn independently from some known distributions $F_1,\ldots,F_n$.
The performance of the decision-maker is measured against a ``prophet'' who knows all the realizations of the rewards in advance, and thus always selects the maximum reward.  
We model our research question using the prophet secretary framework and ask how many additional copies of the rewards are needed to approximate the maximum reward (without the additional copies).
This result implies competition complexity results for both approximating the social welfare, as well as approximating the expected revenue (under the usual regularity assumption) of the optimal offline mechanism (i.e., the second-price auction for social welfare, and the Myerson auction for revenue).

\subsection{Our Contribution}
    In this work, we characterize the asymptotic behavior of the competition complexity of the prophet secretary problem for different classes of algorithms. We consider four natural classes of algorithms for this problem in increasing levels of generality, namely: single-threshold, time-based threshold, activation-based, and general algorithms.
    
     The first class (single-threshold), considers algorithms that set a single-threshold a priori and select the first reward that exceeds it. This type of algorithms is well-studied and widely used in practice.
     The second class (time-based threshold) considers algorithms that can adapt the threshold they use based on the arrival time (but not on the identities of the current reward, nor the set of rewards that arrived so far). 
     The third class (activation-based) considers algorithms that decide (randomly) whether to activate a reward based on its identity, value and arrival time (but independent of the realizations and arrival times of the other rewards). The selected reward is the first activated reward (if such exists). 
     The last class (general) considers algorithms that can base their decisions on all available information including the identity of the current reward, as well as the identities of the rewards arriving so far.
     All previous approaches for the prophet secretary \cite{esfandiari2017prophet,ehsani2018prophet,azar2018prophet,correa2021prophet,harb2023fishing} used algorithms from the first two classes with the exceptions of \citep{correa2017posted} that provided a $(1-\nicefrac{1}{e})$-approximation using a non-adaptive personalized threshold algorithm that is a special case of the third class\footnote{Non-adaptive personalized threshold algorithms are algorithms that set a static threshold for each reward a priori (which is possibly different for different rewards). Then, they accept the first reward exceeding its own threshold.}, 
      and the recent work of \citet{chen2024prophet} that introduced the class of activation-based algorithms, and used it to obtain the currently best known competitive-ratio of $0.688$ for the prophet secretary problem.  

    Our main results (along with the known  competition complexity results of previous works) are summarized in the following table:

\begin{table}[h!]
    \centering
    \begin{tabular}{|c|c|c|c|c|}
        \hline
         & I.I.D. & Block Arrival & Adversarial Arrival & Random Arrival  \\
        \hline
        Single-Threshold & $\Theta\left( \ln\left( \nicefrac{1}{\epsilon}\right)\right)$ \citep{DBLP:journals/corr/abs-2402-11084} & $\Theta\left( \ln\left( \nicefrac{1}{\epsilon}\right)\right)$ \citep{DBLP:journals/corr/abs-2402-11084} & $\Theta\left(  \nicefrac{1}{\epsilon}\right)$ \citep{DBLP:journals/corr/abs-2402-11084} & $\Theta\left(\ln(\nicefrac{1}{\epsilon})\right)$ Thm.~\ref{thm:blind-threshold} \\
        \hline
        Time-Based Threshold &  &  &  & $\Theta\left(\frac{\ln(\nicefrac{1}{\epsilon})}{\ln\ln(\nicefrac{1}{\epsilon})}\right)$ Thm.~\ref{thm:blind} \\
        \cline{1-1} \cline{5-5}
        Activation-Based & \multirow{1}{*}{$\Theta\left( \ln\ln\left( \nicefrac{1}{\epsilon}\right)\right)$ \citep{brustle2023competition}} & \multirow{1}{*}{$\Theta\left( \ln\ln\left( \nicefrac{1}{\epsilon}\right)\right)$ \citep{DBLP:journals/corr/abs-2402-11084}} & \multirow{1}{*}{$\Theta\left(  \nicefrac{1}{\epsilon}\right)$ \citep{DBLP:journals/corr/abs-2402-11084}} & $\Theta\left(\frac{\ln(\nicefrac{1}{\epsilon})}{\ln\ln(\nicefrac{1}{\epsilon})}\right)$ Thm.~\ref{thm:active-lower} \\ \cline{1-1} \cline{5-5}
        General &  &  &  & $\Theta\left(\sqrt{\ln(\nicefrac{1}{\epsilon})}\right)$ Thm.~\ref{thm:general} \\
        \hline
    \end{tabular}
    \caption{Competition complexity results of different classes of algorithms for the different arrival models. We note that time-based threshold algorithms (and also activation-based algorithms) are as powerful as general algorithms for the models of i.i.d., block arrival, and adversarial arrival. This is because, in these models, the optimal algorithms are computed via backward induction, and can be interpreted as time-based threshold algorithms.}
\end{table}

We note that all of our main competition complexity results are asymptotically tight and are supported by explicit constructions. For positive results (upper bounds) that means explicit descriptions of efficient algorithms, and for negative results (lower bounds) that means explicit constructions of hard instances.

\paragraph{Implications of our results.} Our results reveal that while single-threshold algorithms in the random arrival model have the same asymptotic performance as in the i.i.d., and block arrival models, this is no longer true for general algorithms that perform strictly worse in the random arrival model than in the i.i.d., and block arrival models. This implies that the random arrival model is (perhaps surprisingly) harder than the block arrival model. 

Second, our results for time-based threshold and activation-based algorithms reveal that these classes are far from being optimal in terms of competition complexity. This hints that these classes of algorithms may not yield the optimal competitive-ratio in classic settings (i.e., without adding competition). Thus, previous approaches for the prophet secretary problem are likely to be far from optimal.

    \subsection{Our Techniques}
  Our analysis is based on examining the prophet secretary problem within a continuous-time model.
    In the continuous-time model, instead of having an arrival order defined by a random permutation, our arrival order defines for each reward an independent time that is distributed uniformly on the interval $[0,1]$.
    The independence of the arrival times simplify the analysis, and allows us to use several bounds developed in \cite{correa2021prophet} on the probability of not selecting a value by time $t$ for time-based threshold algorithms. 
    In particular, \citet{correa2021prophet} show that for any fixed time-based threshold algorithm, if one replaces two random rewards $V_1,V_2$ by a single reward that is distributed as their maximum, it decreases the probability of stopping by time $t$ for any  $t\in[0,1]$; and if replacing $V_1,V_2$ by two random rewards $V_1',V_2'$ that have the same distribution while maintaining the distribution of the overall maximum, it increases the probability of stopping before time $t$ for any $t\in[0,1]$.
    We utilize these facts to analyze the first two types of algorithms.
    The second ingredient we use for upper bounding the competition complexity is a stochastic domination argument. That is, for every value $x\in\reals_{\geq 0}$, we compare the probabilities that the prophet and our algorithms select a value above $x$.

      \paragraph{Single-threshold algorithms.} For single-threshold algorithms, it is sufficient to consider two types of values of $x$, values above and below the threshold. For the former, we show that a sufficient condition for the stochastic domination argument is that the expected number of times we observe each type of rewards is at least $1$. 
    For the latter, it is sufficient to show that the probability that our algorithm selects a value overall is sufficiently large.

      \paragraph{Time-based threshold  algorithms.} For time-based threshold algorithms, we use a decreasing threshold function that at the beginning (for time $t\in[0,\nicefrac{2}{k}]$) uses the same threshold of the single-threshold algorithm we analyze. Whenever the expected number of times we observe each type of rewards is at least $1$ (which is the sufficient condition for the first case of the analysis of single-threshold algorithms) we can start decreasing the threshold. We decrease the threshold according to the quantiles of the prophet at the rate of $\frac{1}{tk}$, where $t$ is the time, and $k$ is the number of copies of each reward. That is, at time $t$ we use a threshold $\tau(t)$ such that the probability of the prophet being below $\tau(t)$ is $\frac{1}{tk}$. This rate is (up to constant) tight in order to maintain the stochastic dominance relation to the value of the prophet, for each value above $\tau(t)$. Lastly, for values below the last used threshold that correspond to quantile $\frac{1}{k}$ of the prophet, we analyze the probability of not selecting a value overall.

    \paragraph{General algorithms.} Our most technical proof is analyzing the class of general algorithms, where all previous algorithms for the prophet secretary problem developed for competitive-ratio analysis used only non-adaptive strategies (either time-based threshold algorithms, or activation-based algorithms).  
    Our algorithm uses only two thresholds: $\tau_1$ which corresponds to a constant-quantile of the prophet, and $\tau_2$ which corresponds to a $e^{-\sqrt{\ln(1/\epsilon)}}$-quantile of the prophet. The adaptivity plays the role of deciding when to switch between the two thresholds. To set the switching time we define the function $q(t)$ to be the probability of all the rewards from time $t$ (included) having a value below the threshold $\tau_2$. Then the switching time is the last time where $q(t) \leq \epsilon$.
    This adaptive switching time guarantees that the probability of selecting a value overall is at least $1-\epsilon$, which prove the stochastic dominance condition for the case of $x < \tau_2$.
    Also, we show that with high enough probability, the switching time between the two thresholds is late enough, such that the analysis of $x>\tau_1$ is similar to the single-threshold case, and the probability of selecting a value before the switching time (which ensures a value of at least $\tau_1$) is larger than the probability of the prophet having a value above $\tau_2$ which resolves the remaining case of $x\in[\tau_2,\tau_1]$.
Our careful choice for $\tau_1$ and $\tau_2$ combined with the adaptive switching time yield the optimal tradeoff for the different cases in our analysis.

%% file: new-related.tex
\subsection{Further Related Work}
\textbf{Competition complexity.}
Since the work of \citet{BulowKlemperer96} that introduced the competition complexity measure focusing on the single-item case, it has been extended through several dimensions of complexity.
\citet{eden2017competition} and \citet{beyhaghi2019optimal} studied the multi-item setting under additive valuations over independent items, and analyzed the additional amount of bidders necessary for the VCG mechanism to achieve at least the same revenue as the optimal (prior-dependent, computationally unbounded) mechanism as a function of the number of items, and the number of agents. \citet{eden2017competition} devised a lower bound of $\log(m)$ and an upper bound of $n + 2m - 2$ where $n$ is the number of bidders in the original instance, and $m$ is the number of items. \citet{beyhaghi2019optimal} improved the upper bound to $\min\{n(\ln(1 + \nicefrac{m}{n}) + 2),9 \sqrt{nm}\}$. 
\citet{derakhshan2024settling} settled the competition complexity for the case where the number of items is smaller than the number of bidders, proving that the competition complexity is $\Theta(\sqrt{nm})$.
\citet{feldman201899} initiated the study of combining the competition complexity measure with approximation algorithms, showing that near-optimal revenue can be achieved with simple mechanisms (either selling separately, or selling all items as a single bundle) for several settings with a logarithmic or constant multiplicative factor of additional bidders.
\citet{cai202199} improved the result of \citet{feldman201899} showing that an $O(n)$ additional number of bidders is always suffice for obtaining a near-optimal revenue for additive valuations over independent items. \citet{fu2019vickrey} demonstrated that for obtaining a constant approximation of the optimal revenue, it is suffice to add only a single bidder from one of the distributions of the original instance and apply the VCG mechanism.

\citet{liu2018competition} studied the competition complexity in dynamic settings, finding that at most $3n$ additional bidders are required to achieve revenue comparable to the optimal dynamic auction. For two-sided markets, \citet{babaioff2020bulow} developed a budget-balanced, truthful, and prior-independent mechanism that approximates the optimal welfare by adding additional buyers under stochastic dominance assumptions.

Lastly, for the competition complexity of the prophet inequality framework, beyond the works of \citet{brustle2023competition,DBLP:journals/corr/abs-2402-11084} that were mentioned earlier, \citet{ezra2024competition} extended this framework to instances with correlations between rewards in the block and adversarial arrival models, where different copies of rewards are independent.

\textbf{Prophet secretary.}
The prophet secretary problem has been studied extensively since the work of \citet{esfandiari2017prophet} that initiated it.  \citet{correa2017posted} developed a $(1 - \nicefrac{1}{e})$-competitive algorithm using personalized thresholds in a non-adaptive mechanism.
\citet{ehsani2018prophet} showed that the approximation of $1 - \nicefrac{1}{e}$ can be achieved with a single-threshold, and also considered the multi-choice case under matroid feasibility constraints, and combinatorial auctions for which they provide the same guarantee of $1 - \nicefrac{1}{e}$. 
\citet{azar2018prophet} were the first to surpass the $1 - \nicefrac{1}{e}$ barrier with a refined case analysis, developing three time-based threshold algorithms, where the best of the three algorithms achieves a competitive-ratio of $0.6346$. 
\citet{correa2021prophet} introduced the sub-class of time-based threshold algorithms which they termed \textit{blind algorithms}, that use a non-increasing sequence of thresholds based on the distribution of the maximum, and used this class to devise a $0.669$-competitive algorithm. \citet{harb2023fishing} improved the competitive-ratio to $ 0.6724$ using a blind algorithm. Recently, \citet{chen2024prophet} introduced the class of activation-based algorithms and used it to obtain a $0.688$-competitive algorithm surpassing the $0.675$ upper bound for the class of blind algorithms shown by \citet{correa2021prophet}.

In terms of upper bounds, \citet{bubna2023prophet} showed a separation of the competitive-ratios achievable for the prophet secretary problem, and the free-order prophet inequality (where the decision-maker can choose the order of arrival). The best known upper bound on the competitive-ratio for the prophet secretary problem is $0.7235$ \citep{giambartolomei2023prophet}.

Beyond the single-choice case, \citet{arnosti2023tight} 
focused on settings where the algorithm (and the prophet) are allowed to select at most $k$ rewards, and provided a tight constant guarantee for the class of single-threshold algorithms of $\gamma_k \approx 1 - \frac{1}{\sqrt{2 \pi k}}$.
\citet{lee2018optimal} introduced random order contention resolution schemes (RCRS) which is a type of rounding schemes for online algorithms in random arrival models. They used this type of algorithms to provide a $(1-\nicefrac{1}{e})$-approximation algorithm for the prophet secretary problem under matroid feasibility constraints.
\citet{fu2021random} developed a $\frac{8}{15}$-approximation algorithm for the prophet secretary problem for random vertex arrival through the RCRS approach, which was improved by
\citet{macrury2024random} to a $0.535$-approximation for general graphs, and to a tight $0.567$-approximation for bipartite graphs.

Beyond competitive analysis (when comparing to the prophet),  \citet{dutting2023prophet} developed a polynomial-time approximation scheme (PTAS) that approximate the optimal online algorithm (that is not bounded computationally).
Lastly, \citet{cristi2024prophet} considered the sample complexity of the prophet inequality framework. They showed that a constant number of samples per distribution suffices to achieve near optimal competitive-ratios for different variants of the prophet inequality problem, including prophet secretary.

\subsection{Organization}
The paper is organized as follows. We formally state our model in Section~\ref{sec:model} where in Section~\ref{sec:prelim} we present known facts from previous works. In Section~\ref{sec:single} we provide our result for the class of single-threshold algorithms. In Section~\ref{sec:time} we provide our result for the class of time-based threshold algorithms. In Section~\ref{sec:general} we provide our result for the class of general algorithms. Lastly, in Appendix~\ref{app:ex-activation} we provide our result for the class of activation-based algorithms.

%% file: model.tex
\section{Model}\label{sec:model}

\paragraph{Prophet Secretary.} Consider the problem where a decision-maker faces a sequence of $n$ independent rewards $V_1,\ldots,V_n$ that arrive online in a random order and needs to select one of them in an online fashion (not selecting a value is equivalent to receiving a reward of $0$). Each reward $V_i$ is distributed according to some known distribution $F_i\in \Delta$ where $\Delta$ is the set of all distributions over non-negative reals. We denote by $\Delta^*$ the set of all finite product distributions of non-negative reals, i.e.,  $ \Delta^*= \bigcup_{n \in \naturals}\Delta^n$.
The distributions $F_1,\ldots,F_n$ are independent, and we denote by $F$ their product distribution. 
The rewards arrive in a random order $\sigma \sim \permutations{n}$, where $\permutations{n}$ is the uniform distribution over permutations of size $n$. At each time $i=1,\ldots,n$, the algorithm observes $\sigma(i)$, and $V_{\sigma(i)}$, and needs to decide immediately and irrevocably whether to accept it or discard it. If the algorithm selects at time $i$, the process terminates, and the algorithm's utility is $V_{\sigma(i)}$. 
Traditionally, the performance of the decision-maker is measured using the competitive-ratio benchmark, which is the ratio between the expected value chosen by the decision-maker (who makes decisions online) and the expected value of the maximum in hindsight (also referred to as the prophet as he knows the realizations of future rewards, thus always selecting the maximum). We denote an online algorithm for the decision-maker by $\ALG$, and the (possibly random) value chosen by $\ALG$ on a realized sequence defined by $V_1,\ldots,V_n,\sigma$ by $\ALG(V_1,\ldots,V_n,\sigma)$.
Thus, the competitive-ratio of an algorithm $\ALG$ is $$ \rho(\ALG) = \inf_{ F \in  \Delta^*} \frac{\E[\ALG(V_1,\ldots,V_n,\sigma)]}{\E[\max_{i \in [n]} V_i]},$$
where the expectation of the numerator is over the randomness of the rewards $V_1,\ldots,V_n$, over the arrival order $\sigma$, and might also depend on the randomness of the algorithm if such exists.

\paragraph{Continuous-Time Model.} Throughout the paper, we consider an equivalent model of the prophet secretary problem, known as the continuous-time model.
In this model, instead of drawing a random permutation $\sigma \sim \permutations{n}$, each random variable $V_i$ is associated with an arrival time $t_i$, where the $t_i$'s are independent, and are distributed uniformly over the interval $[0,1]$.
The algorithm observes the variables according to their arrival times, and upon an arrival, the algorithm knows the time $t$ (which is equal to some $t_i$), 
and observes the value $V_i$, and its identity $i$. 
Any algorithm for the continuous-time model can be transformed into a randomized algorithm for the discrete-time model by drawing a priori $n$ times $t_1',\ldots,t_n'$  and ordering them, then in the online process when $V_{\sigma(i)}$ arrives, treat it as if  $V_{\sigma(i)}$ arrives at time that is equal to the ${\sigma(i)}$-th smallest time in the multi-set $\{t_1',\ldots,t_n'\}$. 
For the other direction, an algorithm for the discrete-time model can be transformed into an algorithm for the continuous-time model by ignoring the time arrival, and only use as the discrete-time the number of rewards arriving so far. 
It is straightforward to verify, that the output of the algorithms defined by these reductions is distributed the same as the outputs of the original algorithms, thus leading to the same guarantees (since the distribution of the benchmark, i.e., the value of the prophet, is not influenced by the order of arrivals, or the arrival times).

\paragraph{Competition Complexity.}  In this paper, we study the expected performance of an online algorithm that observes more rewards compared to the maximum in hindsight (without the added rewards).
In particular, we analyze algorithms that observe $k$ independent copies of each reward in the original instance. 
Our research question is how many copies are needed to achieve a certain approximation of the prophet's value in the original instance. 
Formally, we consider algorithms that observe a sequence of $N=n \cdot k$ independent random rewards $V_{i,j}$ for $i\in[n]$, $j\in [k]$. Each such variable $V_{i,j}$ is distributed independently from $F_i$, and is associated with an arrival time $t_{i,j}$ distributed uniformly from the interval $[0,1]$. The algorithm $\ALG$ observes the $N$ rewards according to the continuous-time model, and needs to choose one of them in an online fashion.
The goal is to design an algorithm with as few as possible copies that approximates the expected maximum value of the original instance.
We denote by $\opt$ the distribution of the maximum value of the original instance (without the additional copies), i.e., $\opt= \max_{i\in [n]}V_i$.
We denote by $\vec{V}$ the vector of random variables $(V_{1,1},\ldots,V_{n,k})$, and by $\vec{t}$ the corresponding vector of random arrival times $(t_{1,1},\ldots,t_{n,k})$.
Given $\vec{V}$, $\vec{t}$ and some online algorithm $\ALG$, we denote by $\ALG(\vec{V},\vec{t})$ the (possibly random) output of $\ALG$, when running on online input defined by $(\vec{V},\vec{t})$. 
We use the following definition to measure the competition complexity of different classes of algorithms:
\begin{definition}[Competition Complexity]\label{def:competition} Given a class of algorithms $\class$, and $\epsilon \in [0,\nicefrac{1}{e}]$, the $(1-\epsilon)$-competition complexity of class of algorithms $\class$ is the smallest positive integer $k$ such that for every $F\in \Delta^*$, there exists an algorithm $\ALG\in \class$ such that
\[
\E[\ALG(\vec{V},\vec{t})] \geq (1-\epsilon) \cdot \E[\opt].
\]
\end{definition}
We note that we define the competition complexity only for $\epsilon\leq \nicefrac{1}{e}$ since a $(1-\nicefrac{1}{e})$-approximation can be achieved even when using a simple single-threshold algorithm without additional copies \cite{ehsani2018prophet}.
Moreover, the case of $\epsilon = 0$ is also referred to as \emph{exact} competition complexity; it was shown in 
\cite[Theorem 1]{brustle2023competition} that the exact competition complexity is unbounded even for the i.i.d.~case (namely, when $F_1=\ldots=F_n $). So, naturally, our focus will be on the case where $\epsilon > 0$.

\paragraph{Classes of Algorithms.} In this paper we consider the following classes of algorithms (in increasing levels of generality):
\begin{itemize}
    \item \textbf{Single-threshold algorithms:} Given a threshold $\tau\in \reals_{\geq 0 }$, $\ALG_{\tau}$ is the algorithm that returns the first value exceeding $\tau$.
        \item \textbf{Time-based threshold algorithms:} Given a threshold function that maps times to thresholds  $\tau:[0,1] \rightarrow \reals_{\geq 0 }$, $\ALG_{\tau}$ is the algorithm that returns the first value satisfying $V_{i,j} > \tau(t_{i,j})$.
        \item \textbf{Activation-based  algorithms:} Such an algorithm is defined by activation functions $g_{i,j}^v:[0,1]\rightarrow [0,1] $ parameterized by $i\in [n]$, $j\in[k]$ and $v$ in the support of reward $V_{i,j}$, where $g_{i,j}^v$ maps a time $t\in [0,1]$ to a probability of activating reward $V_{i,j}$ when it realizes to have a value of $v$ and arrives at time $t$. The algorithm returns the first reward that is activated.
        \item \textbf{General algorithms:} Algorithms that can base their decisions on all available information, including, most importantly, the set of rewards that arrived so far.
\end{itemize}
We note that single-threshold algorithms behave exactly the same whether they are in the discrete-time model, or in the continuous-time model. Regarding time-based threshold and activation-based algorithms, there is an analog for them in the discrete-time model which uses functions, with a discrete domain of time of $[N]$ instead of having a continuous domain of $[0,1]$. Our tight analysis for the continuous-time model for these classes can be transformed into tight analysis for these classes in the discrete-time model (see Remark~\ref{rem:adapt}).

We note that discrete time-based threshold algorithms are optimal for the prophet inequality setting under known arrival order, and all of the algorithms analyzed so far for obtaining competitive-ratios for the prophet secretary problem \citep{esfandiari2017prophet,ehsani2018prophet,azar2018prophet,correa2021prophet,arnosti2023tight} are (continuous or discrete) time-based threshold algorithms with the exceptions of \citep{correa2017posted} that analyzed the class of non-adaptive personalized thresholds (which is a special case of activation-based  algorithms) and \citep{chen2024prophet} that used an activation-based algorithm. The class of non-adaptive personalized threshold algorithms assigns before the arrival of the rewards a non-adaptive different threshold for each reward. Then it selects the first reward exceeding its threshold. This is a special case of activation-based algorithms, that use $g_{i,j}^v$ that are independent of the time, are always $0$ or $1$, and are $1$ for values of $v$ above some  threshold $\tau_{i,j}$ (and $0$ for values of $v$ below $\tau_{i,j}$).
The class of time-based threshold algorithms is a subclass of the class of activation-based algorithms where all $g_{i,j}^v(t)= \indic{v>\tau(t)}$ for every $i,j,v,t$.

\citet{correa2021prophet}  introduced a subclass of time-based threshold algorithms called \textit{blind  algorithms}.
Each such sub-class is parametrized by a quantile function $\alpha:[0,1]\rightarrow [0,1)$. Then, for every instance defined by $F_1,\ldots,F_n$, it uses the time-based threshold algorithm with threshold function $\tau:[0,1]\rightarrow\reals_{\geq 0}$ defined as $$\tau(t) =\max\{x \mid \Pr[\opt\leq x] \leq \alpha(t)\}.$$ 
A single-threshold algorithm is said to be blind if it uses a threshold of a specific quantile from the distribution of the maximum (independent from the instance). A single-threshold blind algorithm is equivalent to using a constant function $\alpha(t)=c$ for all $t\in[0,1]$ for some constant $c\in[0,1)$.

\paragraph{Simplifying Assumptions and Notations.}
Throughout the paper, we assume for devising our positive results that all of the distributions are continuous. 
This assumption can be made without loss by allowing the algorithms to use randomization to break ties using standard techniques (see, e.g., \citep{ehsani2018prophet}).
This assumption implies that in the definition of the time-based threshold strategy $\tau(t)$ that corresponds to a blind algorithm with a quantile function $\alpha:[0,1]\rightarrow [0,1)$, it holds that $\Pr[\opt\leq \tau(t)] = \alpha(t)$ instead of only $\Pr[\opt\leq \tau(t)] \leq  \alpha(t)$. 
Our constructions for hardness results use point-masses, however, this is not a key property of our constructions, and the examples can be adapted to the continuous case. 

We abuse notation and use $F_i$ to denote both the distribution and its CDF. For two distributions $D_1,D_2$, we denote by $D_1 \cdot D_2$, the distribution of the maximum between two independent samples from $D_1$, and $D_2$ (which has a CDF of the product of the CDF's of $D_1$, and $D_2$).

We use $\sqrt[n]{D}$ as the distribution with a CDF equal to the $n$-th root of the CDF of $D$.

Given some algorithm $\ALG$, we denote by $T$ the stopping time of $\ALG$, where if the algorithm does not select a value, then $T=1$.

Throughout our analysis, we will assume for simplicity that all the arrival times $t_{1,1},\ldots,t_{n,k}$ are different and are strictly smaller than $1$, which happens with high probability, and therefore does not influence the results presented in this paper.

Lastly, when it is clear from context, we omit the input of the algorithm and write $\ALG$ instead of $\ALG(\vec{V},\vec{t})$.

\paragraph{Probability Tools.}
In this paper, we use the following variants of well-known theorems in probability:

\begin{theorem}[Stirling Approximation] For every $n$ it holds that
$$\sqrt{2\pi n} \left(\frac{n}{e}\right)^n \leq n! \leq  e \sqrt{ n} \left(\frac{n}{e}\right)^n .$$
\end{theorem}

\begin{theorem}[Hoeffding's Inequality]
Let $X_1,\ldots,X_n$ be independent random variables where $X_i\in [a_i,b_i]$ almost surely, and $\E[\sum_{i\in [n]} X_i] = \mu$. Then for every $t>0$ it holds that 
$$ \Pr\left[\sum_{i} X_i \geq t +\mu\right] \leq exp\left(-\frac{2t^2}{\sum_{i\in[n]} (b_i-a_i)^2}\right).$$
\end{theorem}

%% file: secretary.tex
\subsection{Preliminaries }\label{sec:prelim}
In this section, we present known facts about time-based threshold algorithms, which variants of them were proven in \citep{correa2021prophet}. In Appendix~\ref{app:prelim}, we prove them for completeness, and since we use slightly different versions than the ones proven in \citep{correa2021prophet}.

Given $F_1,\ldots,F_n$, and $t\in[0,1]$, and a non-increasing function $\tau: [0,1]  \rightarrow R$, let $p_\tau(t,F_1,\ldots,F_n)$ be the probability that the algorithm defined by $\tau$  stops before time $t$.
I.e., for a single distribution it holds that $$ p_{\tau}(t,F_1)  = 
\Pr[t_1< t \wedge  V_1 > \tau(t_1) ] = 
\E[\indic{t_1<t}(1-F_1(\tau(t_1))],$$
and for multiple distributions, it holds that 
$$  p_{\tau}(t,F_1,\ldots,F_n)  = 1- \prod_{i=1}^n (1- p_{\tau}(t,F_i)).$$

\begin{restatable}{lem}{lemff}
\label{lem:f1f23}
For every threshold function $\tau:[0,1]\rightarrow \reals_{\geq 0}$,  and distributions $F_1,F_2$, it holds that  
    $$p_{\tau}(t,F_1\cdot F_2) \leq p_{\tau}(t,F_1,F_2).$$ 
    \end{restatable}

\begin{restatable}{lem}{lemfff} \label{lem:f1f2}
    For every threshold function $\tau:[0,1]\rightarrow \reals_{\geq 0}$, and distributions $F_1,F_2$,  it holds  that  \begin{equation}    
    p_{\tau}(t,F_1,F_2) \leq p_{\tau}(t,\sqrt{F_1\cdot F_2},\sqrt{F_1\cdot F_2}) .  \label{eq:f1f2}
    \end{equation}
\end{restatable}
By applying Lemma~\ref{lem:f1f2} infinitely many times, we get the following corollary:
\begin{corollary}\label{cor:f1fn}
    For every threshold function $\tau:[0,1]\rightarrow \reals_{\geq 0}$, and distributions $F_1,\ldots,F_n$,  it holds  that  
    \begin{eqnarray}    
    p_{\tau}(t,F_1,\ldots,F_n) \leq p_{\tau}(t,\underbrace{\sqrt[n]{F_1\cdot\ldots\cdot F_n},\ldots,\sqrt[n]{F_1\cdot\ldots\cdot F_n}}_{n \text{ times}}) . \nonumber 
    \end{eqnarray}
\end{corollary}

We next observe that for time $t$, the event that some reward $V_i$ arrives at time $t_i=t$ increases the probability of reaching time $t$. 
\begin{observation} \label{obs:tt}
    For every threshold function $\tau:[0,1]\rightarrow \reals_{\geq 0}$,   every instance defined by $F_1,\ldots,F_n$, for every $i\in [n]$, and $ t\in [0,1)$, it holds that $$\Pr[T \geq t_i \mid t_i = t] \geq  \Pr[T \geq t].$$
\end{observation}

Lastly, we observe that we can assume without loss of generality that the threshold function $\tau$ is non-increasing since if one uses the non-increasing ``sorted'' version of $\tau$, then a simple exchange argument implies that the ``sorted'' threshold algorithm leads to at least the same performance.
\begin{observation}
    \label{obs:monotone}
    For every threshold function $\tau:[0,1]\rightarrow\reals_{\geq 0}$, there exists a non-increasing threshold function $\tau':[0,1]\rightarrow\reals_{\geq 0}$, for which for every instance it holds that $\E[\ALG_{\tau'}] \geq \E[\ALG_{\tau}]$.
\end{observation}

%% file: single-threshold.tex
\section{Single-Threshold Algorithms} \label{sec:single}
In this section, we analyze the competition complexity of the class of single-threshold algorithms.
\begin{theorem}
The $(1-\epsilon)$-competition complexity of the class of single-threshold algorithms is $\Theta\left(\ln(\nicefrac{1}{\epsilon})\right)$. 
Moreover, this   
can be achieved using a blind single-threshold algorithm.\label{thm:blind-threshold}
\end{theorem}
\begin{proof}
The proof of the lower bound follows by a lower bound given in \citep{DBLP:journals/corr/abs-2402-11084} for the competition complexity of single-threshold algorithms in the i.i.d. case (where $F_1 =\ldots= F_n$).     

We show that $k = \lceil 2\ln(1/\epsilon)\rceil$ is sufficient.
In the remainder, we present the upper bound.
Consider the threshold $\tau$ satisfying  \begin{equation}
\Pr[\opt \leq \tau]=\frac{1}{2}. \nonumber
\end{equation}
We will use the observations from Section~\ref{sec:prelim}, and thus treat the threshold $\tau$ as a constant threshold function $\tau:[0,1] \rightarrow \reals_{\geq 0}$, where $\tau(t)=\tau$ for every $t\in[0,1]$.

We next show that for every value of $x$, it holds that $$ \Pr[\ALG >x] \geq (1-\epsilon) \Pr[\opt >x],$$
which implies the theorem.

\paragraph{Case 1: $x>\tau$.}  
Let $p_i = \Pr[V_i \leq x]$. 
By the Union bound, it holds that 
\begin{equation}
    \Pr[\opt > x]  \leq \sum_{i \in [n]}\Pr[V_i>x] = \sum_{i\in [n]} (1-p_i). \label{eq:opt2-continuous}
\end{equation}
On the other hand, we note that since $x>\tau$,  if the algorithm reaches time $t_{i,j}$ where $V_{i,j}>x$, then the algorithm selects a value larger than $x$.
Thus,  we get
\begin{eqnarray}
\Pr[\ALG > x] &\geq& \int_{0}^{1} \sum_{i \in [n], j\in [k]} \Pr[T \geq t_{i,j} \wedge V_{i,j} > x \mid t_{i,j} = t]dt \nonumber \\
&=& \int_{0}^{1} \sum_{i \in [n], j\in [k]} \Pr[T \geq t_{i,j} \mid t_{i,j} = t] \cdot \Pr[V_{i,j} > x] dt \nonumber \\
&\geq& k \cdot \sum_{i \in [n]} (1 - p_i) \cdot \int_{0}^{1} \Pr[T \geq t] dt \nonumber \\ 
&=& k \cdot \sum_{i \in [n]} (1 - p_i) \cdot \int_{0}^{1} \left( 1 - p_\tau(t,\underbrace{F_1,\ldots,F_1}_{k \text{ times}},\ldots,\underbrace{F_{n},\ldots,F_n}_{k \text{ times}}) \right) dt,  \label{eq:algx} 
\end{eqnarray}
where the first inequality is since if $t_{i,j}$ is reached and  $V_{i,j} > x$, then reward $V_{i,j}$ is selected; the first equality is since the events of time $t_{i,j}$ is reached and $V_{i,j} > x$ are independent when fixing $t_{i,j} = t$; the second inequality holds by Observation~\ref{obs:tt} and by the definition of $p_i$'s; the second equality is by definition of $p_{\tau}$.

We also know that
\begin{eqnarray}
\int_{0}^{1} \left( 1 - p_\tau(t,\underbrace{F_1,\ldots,F_1}_{k \text{ times}},\ldots,\underbrace{F_{n},\ldots,F_n}_{k \text{ times}}) \right) dt & \geq &
 \int_{0}^{1} \left( 1 - p_{\tau}(t,\underbrace{\sqrt[n]{F_1 \cdot \ldots \cdot F_n},\ldots,\sqrt[n]{F_1 \cdot \ldots \cdot F_n}}_{ nk \text{ times}}) \right)dt \nonumber  \\ &=&  \int_{0}^{1} \left( 1 - p_{\tau}(t,\sqrt[n]{F_1 \cdot \ldots \cdot F_n}) \right)^{nk} dt \nonumber \\
 &=&  \int_{0}^{1} \left( 1-t+t \cdot \sqrt[n]{F_1 \cdot \ldots \cdot F_n(\tau(t))}  \right)^{nk} dt \nonumber \\
&=&  \int_{0}^{1} \left( 1-t+t \cdot \sqrt[n]{\frac{1}{2}} \right)^{nk} dt \geq  \int_{0}^{1} 2^{-tk} dt
 \nonumber \\ 
& =&\left.-\frac{2^{-tk}}{\ln(2)\cdot k}  \right|_{0}^{1}  = \frac{1 -2^{-k}}{\ln(2)\cdot k}  \geq \frac{1}{k}, \label{eq:int}
\end{eqnarray}
where the first inequality holds by Corollary~\ref{cor:f1fn}; the first and second equalities are by definition of $p_{\tau}$; the third equality is since for $t\in[0,1],$ it holds that $F_1 \cdot \ldots \cdot F_n(\tau(t)) = \frac{1}{2}$; the second inequality holds for every $n$; the last inequality holds for every $k\geq 2$.

Overall, we get that
$$ \Pr[\ALG >x  ] \stackrel{\eqref{eq:algx}}{\geq} k \cdot \sum_{i \in [n]} (1 - p_i) \cdot \int_{0}^{1} \left( 1 - p_\tau(t,\underbrace{F_1,\ldots,F_1}_{k \text{ times}},\ldots,\underbrace{F_{n},\ldots,F_n}_{k \text{ times}}) \right) dt \stackrel{\eqref{eq:int}}{\geq} \sum_{i \in [n]} (1 - p_i) \stackrel{\eqref{eq:opt2-continuous}}{\geq}  \Pr[\opt > x],$$
which concludes the proof of this case.

\paragraph{Case 2: $x \leq \tau$.} 
In this case, it holds that 
\begin{eqnarray}
\Pr[\ALG > x] &=&  p_\tau(1,\underbrace{F_1,\ldots,F_1}_{k \text{ times}},\ldots,\underbrace{F_{n},\ldots,F_n}_{k \text{ times}})  \geq p_\tau(1,\underbrace{F_1\cdot \ldots \cdot F_n,\ldots,F_1\cdot\ldots\cdot F_{n}}_{k \text{ times}}) \nonumber \\ & = & 1- \left(1-p_\tau(1,F_1\cdot \ldots \cdot F_n) \right)^k =  1- \left(\frac{1}{2}\right)^k \geq 1-\epsilon \geq (1-\epsilon) \cdot\Pr[\opt >x], \nonumber
\end{eqnarray}
where the first inequality is by Lemma~\ref{lem:f1f23}; the third equality is by definition of $\tau$; the second inequality holds for every $\epsilon>0$ by our choice of $k$; and the last inequality is since a probability is bounded by $1$.
This concludes the proof of the theorem.
\end{proof}

%% file: blind-algorithms.tex
\section{Time-Based Threshold Algorithms} \label{sec:time}
In this section, we analyze the case of time-based threshold algorithms. We show that 
\begin{theorem}
The $(1-\epsilon)$-competition complexity of the class of time-based threshold algorithms is $\Theta\left(\frac{\ln(\nicefrac{1}{\epsilon})}{\ln\ln(\nicefrac{1}{\epsilon})}\right)$.
Moreover, 
this can be achieved using a blind algorithm. \label{thm:blind}
\end{theorem}

Our upper bound is based on constructing a blind algorithm, which is given in Section~\ref{sec:alg-non-adapt}, and our lower bound which holds for any time-based threshold algorithm is proved in Section~\ref{sec:ex-non-adapt}.

\subsection{Blind Algorithms}\label{sec:alg-non-adapt}

Let $\alpha: [0,1] \to [0,1)$ be the following quantile function: 
$$
\alpha(t) = \begin{cases} 
\frac{1}{2}, & \text{if } t \leq \frac{2}{k}, \\
\frac{1}{tk}, & \text{if } t > \frac{2}{k}.
\end{cases}
$$
Note that this is a valid quantile function since $\frac{1}{tk}\in [0,1)$ for every $t> \frac{2}{k}$.
The corresponding threshold function of $\alpha$ is $\tau: [0,1] \to \mathbb{R}$ satisfying for every $t\in[0,1]$ that
$$\Pr\left[\opt \leq \tau(t)\right] = \alpha(t).$$

\begin{lemma}
The $(1-\epsilon)$-competition complexity of the blind  algorithm defined by $\alpha$ is $O\left(\frac{\ln(\nicefrac{1}{\epsilon})}{\ln\ln(\nicefrac{1}{\epsilon})}\right)$. \label{lem:blind-upper}
\end{lemma}

\begin{proof}
We show that $k = \left\lceil\frac{2\ln(\nicefrac{1}{\epsilon})}{\ln\ln(\nicefrac{1}{\epsilon})}\right\rceil$ is sufficient.
We next show that for every value of $x$, it holds that $$ \Pr[\ALG >x] \geq (1-\epsilon) \Pr[\opt >x].$$

\paragraph{Case 1: $x$ is such that $\Pr[\opt >x]\leq \frac{1}{2}$.} The proof of this case is almost identical to the proof of Case 1 in the proof of Theorem~\ref{thm:blind-threshold}, where the only difference is that we analyze the performance of the algorithm until time $\frac{2}{k}$, in which it uses the same threshold as in Theorem~\ref{thm:blind-threshold}. For completeness, the proof of this case can be found in Appendix~\ref{app:case1}.

\paragraph{Case 2: $x$ is such that $\Pr[\opt >x] \in (\frac{1}{2},1-\frac{1}{k} ]$.}

Let $ t_x =\frac{1}{k \cdot \Pr[ \opt  \leq x ]} $.
Note that $1\geq t_x>2/k$, and by definition of $\alpha$ and the corresponding threshold function we have that $\tau(t_x) \geq  x$. Thus, it holds that
\begin{eqnarray}
 \Pr[\ALG > x] & \geq &   p_\tau(t_x,\underbrace{F_1,\ldots,F_1}_{k \text{ times}},\ldots,\underbrace{F_{n},\ldots,F_n}_{k \text{ times}}) \nonumber \\  & \geq & p_\tau(t_x,\underbrace{F_1\cdot \ldots \cdot F_n,\ldots,F_1\cdot\ldots\cdot F_{n}}_{k \text{ times}}) \nonumber \\ & = &1- (1-p_\tau(t_x,F_1\cdot \ldots\cdot F_n))^k, \label{eq:algx1}
 \end{eqnarray}
where the first inequality is since if the algorithm selects a value before time $t_x$, it must be at least $x$, and the second inequality is by Lemma~\ref{lem:f1f23}. 
It also holds that
\begin{equation}
    1- p_\tau(t_x,F_1\cdot \ldots \cdot F_n) = \int_{0}^{\frac{2}{k}} \frac{1}{2} dt + \int_{\frac{2}{k}}^{t_x} \frac{1}{tk}dt + \int_{t_x}^1 1dt = 1-t_x+\frac{\ln(ekt_x/2)}{k}   .\label{eq:1minp}
\end{equation} 
Overall, we get that
$$ \Pr[\ALG > x]\geq  1-\left(1-t_x+\frac{\ln(ekt_x/2)}{k} \right)^k  \geq 1- \frac{e\cdot t_x \cdot k}{2} \cdot e^{-t_x \cdot k} \geq 1-\frac{1}{t_x\cdot k}  =  \Pr[\opt > x], $$
where the first inequality is by Inequalities~\eqref{eq:algx1} and \eqref{eq:1minp}, the second inequality is by the identity $e^{-x} \geq (1-x/k)^k $, and the third inequality holds for every value of $t_x>0$.

\paragraph{Case 3: $x$ is such that $\Pr[\opt >x]  >1-\frac{1}{k} $.} 
In this case, for every time $t \in [0,1]$, it holds that 
$ \tau(t) >x$, which implies that the probability of selecting a value larger than $x$, is the probability of selecting a value overall. Thus,
\begin{eqnarray}
\Pr[\ALG > x] &=&  p_\tau(1,\underbrace{F_1,\ldots,F_1}_{k \text{ times}},\ldots,\underbrace{F_{n},\ldots,F_n}_{k \text{ times}})  \geq p_\tau(1,\underbrace{F_1\cdot \ldots \cdot F_n,\ldots,F_1\cdot\ldots\cdot F_{n}}_{k \text{ times}}) \nonumber \\ & = & 1- \left(1-p_\tau(1,F_1\cdot \ldots \cdot F_n) \right)^k =  1- \left(\frac{\ln(ek/2)}{k}\right)^k \geq 1-\epsilon \geq (1-\epsilon) \cdot\Pr[\opt >x], \nonumber
\end{eqnarray}
where the first inequality is by Lemma~\ref{lem:f1f23}; the third equality is by setting $t_x=1$ in Equation~\eqref{eq:1minp}; the second inequality holds for every $\epsilon <\nicefrac{1}{e}$ by our choice of $k$; and the last inequality is since a probability is bounded by $1$.
This concludes the proof of the lemma.
\end{proof}

%% file: non-adapt-hard.tex
\subsection{Time-Based Threshold Algorithms: Hardness}\label{sec:ex-non-adapt}

In this section, we show a tight bound on the competition complexity of time-based threshold algorithms.
\begin{lemma}
The $(1-\epsilon)$-competition complexity of the class of time-based threshold algorithms is $\Omega\left(\frac{\ln(\nicefrac{1}{\epsilon})}{\ln\ln(\nicefrac{1}{\epsilon})}\right)$. \label{lem:blind-lower}
\end{lemma}
\begin{proof}

We show that $ k = \frac{\ln(\nicefrac{1}{\epsilon})}{4\ln\ln(\nicefrac{1}{\epsilon})}$ is not sufficient (we assume for simplicity of the proof that $k$ is an integer). Consider an instance consisting of two types of rewards:
the first has a deterministic value of 1, the second is random and has value $1+\sqrt{\epsilon}$ with probability $1-p$ and value 0 with probability $p$, where $p = \frac{1}{k}$.

The expected value of the prophet is $1+\sqrt{\epsilon}-p\cdot \sqrt{\epsilon}$.
Thus, giving $(1-\epsilon)$-approximation means that the algorithm needs to have an expected reward of at least 
$$ (1-\epsilon)\cdot \E[\opt] = (1-\epsilon)(1+\sqrt{\epsilon}-p\cdot \sqrt{\epsilon}) > 1+\sqrt{\epsilon}-3p\cdot \sqrt{\epsilon} .$$

Since there are only three values in the supports of the distributions of the rewards, and one of them is $0$ (which an optimal algorithm can always discard), then by Observation~\ref{obs:monotone}, we can consider time-based threshold algorithms that use a non-increasing threshold $\ALG_t$ that are defined by a time $t \in [0,1]$, where until time $t$ they use accept only the value of $1+\sqrt{\epsilon}$, and from time $t$ they accept values in $\{1,1+\sqrt{\epsilon}\}$.
 Consider an algorithm $\ALG_t$ for some time $t$. We distinguish between two cases depending on the time $t$.
\paragraph{Case 1: $t\leq \frac{1}{2k}$.} 
In this case, we say that an order of arrival of $N$ rewards is ``bad'' if all rewards arrive after time $t$.
The probability of a ``bad'' order is 
\begin{equation}    
\Pr[ \mbox{``bad'' order} ] = \left(1-t\right)^{2k} \geq \left(1-\frac{1}{2k}\right)^{2k} \geq \frac{1}{4}, \label{eq:bad-non}
\end{equation}
where the last inequality holds for every $k\geq 1$.

Under a ``bad'' order, the algorithm selects the first non-zero value. Since the amount of rewards with a value of $1$ is at least the amount of rewards with a value of $1+\sqrt{\epsilon}$, and since under the event that the order is ``bad'', all of the rewards arrive uniformly at random, we get that 
\begin{equation}
    \Pr[\ALG_t =1 \mid \text{``bad'' order}] \geq \frac{1}{2}. \label{eq:1bad}
\end{equation}
Thus, the performance of  algorithm $\ALG_t$ is bounded by 
$$ \E[\ALG_t] \leq 1+\sqrt{\epsilon}-\sqrt{\epsilon}  \cdot\Pr[\ALG_t =1 \mid \text{``bad'' order} ] \cdot \Pr[ \mbox{``bad'' order} ]\stackrel{\eqref{eq:bad-non},\eqref{eq:1bad}}{\leq} 1+\sqrt{\epsilon}-\sqrt{\epsilon} \cdot \frac{1}{2} \cdot \frac{1}{4} = 1+\sqrt{\epsilon}-\frac{\sqrt{\epsilon}}{8}.$$

For sufficiently small $\epsilon$ (such that $k \geq 25$), it holds that $3p< \frac{1}{8}$, which means that $\ALG$ does not provide a $(1-\epsilon)$-approximation, which concludes the proof of this case.

\paragraph{Case 2: $t>\frac{1}{2k}$.}
In this case, we say that an order of arrival of $N$ rewards is ``bad'' if all of the deterministic rewards (rewards of the first type) arrive before time $t$.

Under ``bad'' order, if all the randomized rewards (rewards of the second type) have a value of $0$, then the algorithm gains a reward of $0$.
Thus, 
\begin{eqnarray}
     \E[\ALG_t]  & \leq & (1+\sqrt{\epsilon})(1-\Pr[\ALG_t =0]) \leq  (1+\sqrt{\epsilon})(1-\Pr[\ALG_t =0 \mid \text{``bad'' order} ] \cdot \Pr[ \mbox{``bad'' order} ] ) \nonumber \\ & \stackrel{}{\leq}& (1+\sqrt{\epsilon})(1 -p^k \cdot t^k    )\leq 1+\sqrt{\epsilon}- \left(\frac{1}{2k^2}\right)^k . \nonumber 
\end{eqnarray}
For sufficiently small $\epsilon$ (such that $k \geq 3$), it holds that $3p\sqrt{\epsilon}< \left(\frac{1}{2k^2}\right)^k $, which means that $\ALG$ does not provide a $(1-\epsilon)$-approximation.
This concludes the proof of the lemma.
\end{proof}

\begin{remark}
We note that Lemma~\ref{lem:blind-lower} can be adapted easily to the discrete-time model. Consider an instance where besides the two rewards defined in the instance described in the lemma, it has $m$ additional rewards that are deterministically 0. This does not change the distribution of the maximum value, however, the times in which the important rewards arrive behave similarly to the continuous-time model.
Again, we can assume that the discrete-time-based thresholds are non-increasing.
When $m$ goes to infinity, the performance of a discrete-time-based threshold algorithm that until time $t \in [(m+2)k]$ uses a threshold of $1+\sqrt{\epsilon}$  and after time $t$ uses a threshold of $1$, converges to the performance of the algorithm $\ALG_{t'}$ for $t'=\frac{t}{(m+2)k}$. 
    \label{rem:adapt} 
\end{remark}

%% file: adaptive.tex
\section{General Algorithms}\label{sec:general}
In this section, we prove our main result:
\begin{theorem}
The $(1-\epsilon)$-competition complexity of the class of general algorithms is $\Theta\left(\sqrt{\ln(\nicefrac{1}{\epsilon})}\right)$. \label{thm:general}
\end{theorem} 
Our upper bound is given in Section~\ref{sec:alg-adapt}, and our lower bound is proved in Section~\ref{sec:ex-adapt}.

\subsection{General Algorithms: Algorithm}\label{sec:alg-adapt}
\begin{lemma}
The $(1-\epsilon)$-competition complexity of the class of general algorithms is $O\left(\sqrt{\ln(\nicefrac{1}{\epsilon})}\right)$.
\end{lemma}
\begin{proof}
We assume for simplicity of the proof that $\epsilon$ is of the form $e^{-\ell^2}$ for some integer $\ell\geq 2$.
We show that $k = 8\sqrt{\ln(\nicefrac{1}{\epsilon})} = 8\ell$ is sufficient.
Let $\tau_1,\tau_2$ be two threshold satisfying that $\Pr\left[\opt \leq \tau_1 \right] = \frac{3}{4}$, and $\Pr\left[\opt \leq \tau_2 \right] = e^{-\ell}$.
Let $q_i= \Pr[V_i \leq \tau_2]$.
We define  $q(t)$ for every $t\in [0,1]$  to be the probability of not selecting a value if starting at time $t$ and using the single-threshold $\tau_2$. I.e., $$ q(t) = \prod_{i\in [n],j\in [k]} \left(1-(1-q_i)\cdot {\indic{t_{i,j}\geq t}}\right) .$$
We note that for every realization of $t_{i,j}$'s, the function of $q(t)$ is weakly increasing in $t$, and can only strictly increase immediately after some time $t_{i,j}$, thus it is right continuous.
We define $S$ to be the random variable of the last time where $q(t)$ is at most $\epsilon$, i.e., $S=\max\{t\in [0,1] \mid q(t) \leq \epsilon\}$ (the maximum is well-defined since $q$ is increasing and right-continuous, with $q(0) = \prod_{i\in[n]}q_i^k = e^{-k\ell} = \epsilon^8 \leq \epsilon$, which implies that the maximum is over a non-empty set).
Note that $S$ is only dependent on the values of $\vec{t}$, and is independent from the values of $\vec{V}$. 

Our algorithm works as follows: 
\begin{enumerate}
    \item For time $t \in [0,S)$, use threshold $\tau_1$.
    \item  For time $t \in [S,1]$, use threshold $\tau_2$.
\end{enumerate}
We note that the algorithm is well-defined, since at time $t$, it can be deduced whether $t<S$ or $t \geq S$ since it is a function of the set of rewards that arrived so far.  

Throughout the proof, for indexes $i,j$ we denote the variable $T_{i,j}$  (respectively, $S_{i,j}$) as the stopping time (respectively, the time of switching thresholds) of our algorithm on an instance without $V_{i,j}$. 
We also define $T'$ as the stopping time of the algorithm that uses only threshold $\tau_1$ for all $t\in[0,1]$.
We define respectively, $T_{i,j}'$ as the stopping time of the single-threshold algorithm (with threshold $\tau_1$) on an instance without $V_{i,j}$. 
\begin{observation}
For every $i\in [n], j\in[k]$, and every realization of $\vec{V}$, $\vec{t}$, it holds that $ T \geq \min\{t_{i,j},T_{i,j}\}$. \label{obs:tij}
\end{observation}
\begin{proof}
For every realization of $\vec{t}$ it holds that $S_{i,j} \leq S$, therefore, when applying our algorithm on the instance without $V_{i,j}$, it always uses a threshold lower or equal to the threshold used in the instance with $V_{i,j}$. Thus, $T<T_{i,j}$ can only happen if $V_{i,j}$ is selected, for which  $T=t_{i,j}$.
\end{proof}

We next show that for every value of $x$, it holds that $$ \Pr[\ALG >x] \geq (1-\epsilon) \Pr[\opt >x].$$

\paragraph{Case 1: $x\geq \tau_1$.}
Let $p_i = \Pr[V_i \leq x]$. 
By the Union bound, it holds that 
\begin{equation}
    \Pr[\opt > x]  \leq \sum_{i \in [n]}\Pr[V_i>x] = \sum_{i\in [n]} (1-p_i). \label{eq:opt3-continuous}
\end{equation}
On the other hand
\begin{eqnarray}
\Pr[\ALG > x] & = & \sum_{i\in[n],j\in[k]} \Pr[V_{i,j}  >x \wedge t_{i,j} \leq  T]  = \sum_{i\in[n],j\in[k]} \Pr[V_{i,j}  >x ] \cdot \Pr[t_{i,j} \leq  T] \nonumber \\    
& \geq  & \sum_{i\in[n],j\in[k]} \Pr[V_{i,j}  >x ] \cdot \Pr[ t_{i,j} \leq \min\{t_{i,j},T_{i,j}\}]  \nonumber \\ 
& =  & \sum_{i\in[n],j\in[k]} \Pr[V_{i,j}  >x] \cdot \Pr[t_{i,j} \leq T_{i,j}] = \sum_{i\in[n],j\in[k]} (1-p_i) \cdot \E[T_{i,j}] , \label{eq:algxc1}
\end{eqnarray}
where the first equality is since if the algorithm reaches some  $t_{i,j}$, and  $V_{i,j}>x\geq \tau_1$, then the algorithm returns this value; the second equality is since whether time $t_{i,j}$ is reached is independent of the value of $V_{i,j}$; 
the inequality is by Observation~\ref{obs:tij}; the last equality is since $t_{i,j}$ and  $T_{i,j}$ are independent and $t_{i,j}$ is distributed uniformly over $[0,1]$.

We can bound the expected value of $T_{i,j}$ by
\begin{equation}
    \E[T_{i,j}]  \geq   \E[\min\{ T_{i,j}', S_{i,j}\}]  \geq  \frac{2}{k} \Pr\left[\min\{ T_{i,j}', S_{i,j}\} \geq \frac{2}{k}\right] \geq \frac{2}{k} \left(  \Pr\left[T_{i,j}' \geq  \frac{2}{k}\right] -\Pr\left[S_{i,j} < \frac{2}{k}\right]  \right)  \label{eq:tij1},
\end{equation}
where the first inequality is since $T_{i,j} < S_{i,j} $ implies that $T_{i,j}=T'_{i,j}$, and the second inequality is by Markov's Inequality.

We next bound separately each of the terms $\Pr\left[T_{i,j}' \geq \frac{2}{k}\right] $ and $\Pr\left[S_{i,j} < \frac{2}{k}\right]$. It holds that  

\begin{eqnarray}
\Pr\left[T_{i,j}'   \geq  \frac{2}{k}\right] & \geq & \Pr\left[T' \geq \frac{2}{k}\right] =1-p_{\tau_1}(\frac{2}{k},\underbrace{F_1,\ldots,F_1}_{k \text{ times}},\ldots,\underbrace{F_{n},\ldots,F_n}_{k \text{ times}})  \nonumber \\ 
& \geq & 1- p_{\tau_1}(\frac{2}{k},\underbrace{\sqrt[n]{F_1 \cdot \ldots \cdot F_n},\ldots,\sqrt[n]{F_1 \cdot \ldots \cdot F_n}}_{ N \text{ times}})  \nonumber \\ 
&=& \left(1-p_{\tau_1}(\frac{2}{k},\sqrt[n]{F_1 \cdot \ldots \cdot F_n})\right)^{nk} = \left(1-\frac{2}{k}\left(1-\sqrt[n]{\frac{3}{4}}\right) \right)^{nk} \geq  \frac{9}{16}, \label{eq:tij2}
\end{eqnarray}
where the first inequality is since for the single-threshold algorithm, adding $V_{i,j}$ at time $t_{i,j}$ can only decrease the stopping time; the second inequality is by Corollary~\ref{cor:f1fn}; the last equality is by the definition of $\tau_1$; and the last inequality holds for every $n,k$.

For bounding $\Pr\left[S_{i,j} < \frac{2}{k}\right]$, we define $w_i = -\ln{\left(q_i\right)}$.
By definition of $\tau_2$, we know that $\Pr[V_i \leq \tau_2] \geq \Pr[\opt \leq \tau_2] >0$, so $w_i$ is well-defined. Moreover, since $\Pr[\opt \leq \tau_2] = \prod_{i \in [n]} \Pr[V_i \leq \tau_2]$, we get that $$\sum_{i\in[n]} w_i =-\sum_{i \in [n]} \ln(\Pr[V_i \leq \tau_2]) = -\ln(\Pr[\opt \leq \tau_2]) = \ell .$$ 
Given a realization of $\vec{t}$, we denote by $A(\vec{t})$ the set of variables $V_{i',j'}$ arriving at time $t_{i',j'} \geq \frac{2}{k}$, i.e., $$A(\vec{t}) = \left\{(i',j')\in [n] \times [k]  \mid t_{i',j'} \geq \frac{2}{k} \right\}.$$
We bound the probability of $S_{i,j} < \frac{2}{k}$ by 
\begin{eqnarray}
    \Pr\left[S_{i,j} < \frac{2}{k}\right]  & = & \Pr\left[
    \prod_{(i',j')\in A(\vec{t}) \setminus \{(i,j)\}}q_{i'} > \epsilon \right] = \Pr\left[\sum_{(i',j') \in A(\vec{t})\setminus \{(i,j)\}} w_{i'}  < \ell^2  \right] \nonumber \\
&\leq&    \Pr\left[\sum_{(i',j') \in A(\vec{t})} w_{i'}  < \ell^2  +w_i \right]
\nonumber \\
    & = &  \Pr\left[\sum_{i'\in[n],j'\in[k]} w_{i'} \cdot \indic{t_{i',j'} \geq \frac{2}{k}} < \ell^2   +w_i\right]
    \nonumber \\ &\leq  &  \Pr\left[\sum_{i'\in[n],j'\in [k]} w_{i'} \cdot \indic{t_{i',j'} < \frac{2}{k}} > 7\ell^2 - \ell \right] \nonumber \\ 
    &  \leq &  exp\left({-\frac{2\left(7\ell^2 -3\ell \right)^2}{\sum_{i'\in[n],j'\in[k]}{w_{i'}^2}}}\right) \leq   
    exp\left({-\frac{2\left(4\ell^2\right)^2}{k\cdot (\sum_{i'\in [n]}w_{i'})^2}}\right) = e^{-4\ell} \leq \frac{1}{16} , \label{eq:sij}
\end{eqnarray}
where the first equality is by the definition of $S_{i,j}$; the second equality is by taking $\log$ over both sides and rearranging; the first inequality is by adding the term corresponding to $(i,j)$ which is bounded by $w_{i}$; the third equality is by the definition of $A(\vec{t})$; the second inequality is since $\sum_{i'\in[n],j'\in[k]}w_{i'} =k\ell = 8\ell^2$, and since $w_i \leq \ell$; the third inequality is by Hoeffding's Inequality (where the term $3\ell$ is the expectation of $\sum_{i'\in[n],j'\in [k]} w_{i'} \cdot \indic{t_{i',j'} < \frac{2}{k}}$); the fourth inequality holds since $\ell\geq 1$ and since $w_{i'}$'s are non-negative; the last inequality holds since $\ell \geq 1$. 

Combining everything together, we get that \begin{eqnarray}
    \Pr[\ALG > x] & \stackrel{\eqref{eq:algxc1}}{\geq}&  \sum_{i\in[n],j\in[k]} (1-p_i) \cdot \E[T_{i,j}] \stackrel{\eqref{eq:tij1}}{\geq } \sum_{i\in[n],j\in[k]} (1-p_i) \cdot \frac{2}{k} \left(  \Pr\left[T_{i,j}' \geq  \frac{2}{k}\right] -\Pr\left[S_{i,j} \leq \frac{2}{k}\right] \right) \nonumber \\ 
    & \stackrel{\eqref{eq:tij2},\eqref{eq:sij}}{\geq} & \sum_{i\in[n],j\in[k]} (1-p_i) \cdot \frac{2}{k} \left(  \frac{9}{16} -\frac{1}{16} \right) = \sum_{i \in [n]} (1-p_i) \stackrel{\eqref{eq:opt3-continuous}}{\geq} \Pr[\opt > x], \nonumber
    \end{eqnarray}
    which concludes the proof of the case.

\paragraph{Case 2: $x\in [\tau_2,\tau_1)$.}
We first bound the probability of $\opt$ selecting a value larger than $x$ by
\begin{equation}
     \Pr[\opt>x] \leq \Pr[\opt>\tau_2] \leq 1-e^{-\ell} . \label{eq:optc2}
     \end{equation} 
On  the other hand
\begin{equation}
    \Pr[\ALG>x] \geq \Pr[T<S] = \Pr[T'<S] \geq \Pr\left[T' < \frac{1}{2} \leq S\right] \geq 1- \Pr\left[T' \geq  \frac{1}{2}\right] - \Pr\left[S < \frac{1}{2}\right] . \label{eq:algxc2}
\end{equation}
We can bound the probability of $T'\geq  \frac{1}{2}$ by
\begin{eqnarray}    
 \Pr\left[T' \geq \frac{1}{2}\right] & =& 1- p_{\tau_1}(\frac{1}{2},\underbrace{F_1,\ldots,F_1}_{k \text{ times}},\ldots,\underbrace{F_{n},\ldots,F_n}_{k \text{ times}}) \leq  1 - p_{\tau_1}(\frac{1}{2},\underbrace{F_1\cdot \ldots \cdot F_n,\ldots,F_1\cdot\ldots\cdot F_{n}}_{k \text{ times}})  \nonumber \\ 
 & = & (1-p_{\tau_1}(\frac{1}{2},F_1\cdot \ldots \cdot F_n))^{k}  = \left(1-\frac{1}{2}\left(1-\frac{3}{4}\right)\right)^{k}  =  \left(\frac{7}{8}\right)^{8\ell} \label{eq:t13},
 \end{eqnarray}
 where the inequality is by Lemma~\ref{lem:f1f23}.

Given a realization of $\vec{t}$, we denote by $A(\vec{t})$ the set of variables $V_{i,j}$ arriving at time $t_{i,j} \geq \frac{1}{2}$, i.e., $$A(\vec{t}) = \left\{(i,j)\in [n] \times [k]  \mid t_{i,j} \geq \frac{1}{2} \right\}.$$
We bound the probability of $S < \frac{1}{2}$ by 
\begin{eqnarray}
    \Pr\left[S < \frac{1}{2}\right]  & = & \Pr\left[
    \prod_{(i,j)\in A(\vec{t})}q_{i} > \epsilon \right] = \Pr\left[\sum_{(i,j) \in A(\vec{t})} w_{i}  < \ell^2  \right] \nonumber \\
    & = &  \Pr\left[\sum_{i\in[n],j\in[k]} w_{i} \cdot \indic{t_{i,j} \geq \frac{1}{2}} < \ell^2  \right]
    \nonumber \\ &= &  \Pr\left[\sum_{i\in[n],j\in [k]} w_{i} \cdot \indic{t_{i,j}  < \frac{1}{2}} > 7\ell^2 \right] \nonumber \\ 
    &  \leq &  exp\left({-\frac{2\left(7\ell^2 - 4\ell^2\right)^2}{\sum_{i'\in[n],j'\in[k]}{w_{i'}^2}}}\right) \leq   
    exp\left({-\frac{2\left(3\ell^2\right)^2}{k\cdot (\sum_{i'\in [n]}w_{i'})^2}}\right) =  e^{-\frac{9}{4}\ell} , \label{eq:sijnew}
\end{eqnarray}
where the first inequality is by Hoeffding's inequality (where the term $4\ell^2$ is the expectation of $\sum_{i\in[n],j\in [k]} w_{i} \cdot \indic{t_{i,j} < \frac{1}{2}}$), and the second inequality is since $w_i$ are non-negative.

Overall, we get that $$\Pr[\ALG>x]  \stackrel{\eqref{eq:algxc2}}{\geq} 1- \Pr\left[T' \geq \frac{1}{2}\right] - \Pr\left[S <  \frac{1}{2}\right]  \stackrel{\eqref{eq:t13},\eqref{eq:sijnew}}{\geq }    1- \left(\frac{7}{8}\right)^{8\ell} - e^{-\frac{9}{4}\ell} \geq  1-e^{-\ell} \stackrel{\eqref{eq:optc2}}{\geq} \Pr[\opt>x], $$
where the third inequality holds for every $\ell \geq 2$. This concludes the proof of this case.

\paragraph{Case 3: $x<\tau_2$.}
It holds that $$ \Pr[\ALG>x] \geq \Pr[\exists{(i,j) \in[n] \times [k]}: V_{i,j} > \tau_2  \wedge t_{i,j} \geq S ] = 1- q(S) \geq  1-\epsilon \geq (1-\epsilon)\Pr[\opt>x] ,$$
where the first inequality is since if there exists such $V_{i,j}$, then a value of at least $\tau_2>x$ is selected; the first equality is by definition of $q(t)$; the second inequality is by definition of $S$; and the last inequality is since a probability is bounded by $1$.
This concludes the proof of the lemma.
\end{proof}

%% file: example.tex
\subsection{General Algorithms: Hardness}\label{sec:ex-adapt}
In this section, we show a lower bound on the competition complexity of general algorithms. We note that this hardness result applies even for algorithms that know the random arrival order in advance. 
\begin{lemma}
The $(1-\epsilon)$-competition complexity of the class of general algorithms is $\Omega\left(\sqrt{\ln(\nicefrac{1}{\epsilon})}\right)$.
\end{lemma}
\begin{proof}
Consider an instance consisting of two types of rewards:
the first has a deterministic value of 1, the second is random and has value $1+\sqrt{\epsilon}$ with probability $1-p$ and value 0 with probability $p$, where $p = e^{-\sqrt{\ln(1/\epsilon)}}$.
We show that $ k = \frac{\sqrt{\ln(1/\epsilon)}}{2}$ is not sufficient.

We say that an order of arrival of $N$ rewards is ``bad'' if all the deterministic rewards (rewards of the first type) arrive before all of the randomized rewards (rewards of the second type).
Under such an order, an online algorithm cannot obtain an expected value of more than $1$.
This is since the expected value of the algorithm when discarding all of the deterministic rewards is  $(1+\sqrt{\epsilon}) \cdot (1-p^k) = (1+\sqrt{\epsilon}) \cdot (1-\sqrt{\epsilon}) = 1-\epsilon <1$ (which implies that it is better to select one of the deterministic rewards).

Under non-bad order, the algorithm can obtain at most $1+a$, which is the maximum value in the support of the distributions.

The probability of a ``bad'' order is 
\begin{equation}    
\Pr[ \mbox{``bad'' order} ] = \frac{(k!)^2}{(2k)!}  \geq \frac{2\pi k \left(\frac{k}{e}\right)^{2k}}{e\sqrt{2k} \left(\frac{2k}{e}\right)^{2k}} \geq \frac{1}{4^k}, \label{eq:bad}
\end{equation}
where the first inequality is by Stirling approximation.

Thus, the performance of any algorithm is bounded by 
$$ \E[\ALG] \leq \Pr[\mbox{``bad'' order} ] \cdot 1 + \Pr[\mbox{not ``bad'' order}] \cdot (1+\sqrt{\epsilon}) \stackrel{\eqref{eq:bad}}{\leq} 1+\sqrt{\epsilon}-\frac{\sqrt{\epsilon}}{4^k}. $$
On the other hand, the expected value of the prophet is $\E[\opt]=1+\sqrt{\epsilon}-p\cdot\sqrt{\epsilon}$.

However, $$ (1-\epsilon) \cdot\E[\opt]=(1-\epsilon)(1+\sqrt{\epsilon}-p\cdot \sqrt{\epsilon})  >  1+\sqrt{\epsilon}-3p\cdot \sqrt{\epsilon} ,$$ while for sufficiently small $\epsilon$ (such that $k\geq 4$), it holds that $3p< \frac{1}{4^k}$, which means that $\ALG$ does not provide a $(1-\epsilon)$-approximation, which concludes the proof.
\end{proof}

%% file: app-case1.tex
\section{Proof of Case 1 of Lemma~\ref{lem:blind-upper}}\label{app:case1}
Let $p_i = \Pr[V_i \leq x]$. 
By the Union bound, it holds that 
\begin{equation}
    \Pr[\opt > x]  \leq \sum_{i \in [n]}\Pr[V_i>x] = \sum_{i\in [n]} (1-p_i). \label{eq:opt2-continuous-app}
\end{equation}
On the other hand, we note that since $x$ is above the median of $\opt$ and since $\tau(t)$ is the median of $\opt$, then $x>\tau(t)$. Thus, if the algorithm reaches time $t_{i,j}$ where $V_{i,j}>x$, then the algorithm selects a value larger than $x$.
Thus, by analyzing the probability that the algorithm selects a value larger than $x$ up to time $\frac{2}{k}$, we get
\begin{eqnarray}
\Pr[\ALG > x] &\geq& \int_{0}^{\frac{2}{k}} \sum_{i \in [n], j\in [k]} \Pr[T \geq t_{i,j} \wedge V_{i,j} > x \mid t_{i,j} = t]dt \nonumber \\
&=& \int_{0}^{\frac{2}{k}} \sum_{i \in [n], j\in [k]} \Pr[T \geq t_{i,j} \mid t_{i,j} = t] \cdot \Pr[V_{i,j} > x] dt \nonumber \\
&\geq& k \cdot \sum_{i \in [n]} (1 - p_i) \cdot \int_{0}^{\frac{2}{k}} \Pr[T \geq t] dt \nonumber \\ 
&=& k \cdot \sum_{i \in [n]} (1 - p_i) \cdot \int_{0}^{\frac{2}{k}} \left( 1 - p_\tau(t,\underbrace{F_1,\ldots,F_1}_{k \text{ times}},\ldots,\underbrace{F_{n},\ldots,F_n}_{k \text{ times}}) \right) dt,  \label{eq:algx-app} 
\end{eqnarray}
where the first inequality is since if $t_{i,j}$ is reached and  $V_{i,j} > x$, then reward $V_{i,j}$ is selected; the first equality is since the events of time $t_{i,j}$ is reached and $V_{i,j} > x$ are independent when fixing $t_{i,j} = t$; the second inequality holds by Observation~\ref{obs:tt} and by the definition of $p_i$'s; the second equality is by definition of $p_{\tau}$.

We also know that
\begin{eqnarray}
\int_{0}^{\frac{2}{k}} \left( 1 - p_\tau(t,\underbrace{F_1,\ldots,F_1}_{k \text{ times}},\ldots,\underbrace{F_{n},\ldots,F_n}_{k \text{ times}}) \right) dt & \geq &
 \int_{0}^{\frac{2}{k}} \left( 1 - p_{\tau}(t,\underbrace{\sqrt[n]{F_1 \cdot \ldots \cdot F_n},\ldots,\sqrt[n]{F_1 \cdot \ldots \cdot F_n}}_{ nk \text{ times}}) \right)dt \nonumber  \\ &=&  \int_{0}^{\frac{2}{k}} \left( 1 - p_{\tau}(t,\sqrt[n]{F_1 \cdot \ldots \cdot F_n}) \right)^{nk} dt \nonumber \\
 &=&  \int_{0}^{\frac{2}{k}} \left( 1-t+t \cdot \sqrt[n]{F_1 \cdot \ldots \cdot F_n(\tau(t))}  \right)^{nk} dt \nonumber \\
&=&  \int_{0}^{\frac{2}{k}} \left( 1-t+t \cdot \sqrt[n]{\frac{1}{2}} \right)^{nk} dt \geq  \int_{0}^{\frac{2}{k}} 2^{-tk} dt
 \nonumber \\ 
& =&\left.-\frac{2^{-tk}}{\ln(2)\cdot k} \right|_{0}^{\frac{2}{k}}  = \frac{1-\frac{1}{4}}{\ln(2) \cdot k}   \geq \frac{1}{k}, \label{eq:int-app}
\end{eqnarray}
where the first inequality holds by Corollary~\ref{cor:f1fn}; the first and second equalities are by definition of $p_{\tau}$; the third equality is since for $t\in[0,\frac{2}{k}],$ it holds that $F_1 \cdot \ldots \cdot F_n(\tau(t)) = \frac{1}{2}$; the second inequality holds for every $n$.

Overall, we get that
$$ \Pr[\ALG >x  ] \stackrel{\eqref{eq:algx-app}}{\geq} k \cdot \sum_{i \in [n]} (1 - p_i) \cdot \int_{0}^{\frac{2}{k}} \left( 1 - p_\tau(t,\underbrace{F_1,\ldots,F_1}_{k \text{ times}},\ldots,\underbrace{F_{n},\ldots,F_n}_{k \text{ times}}) \right) dt \stackrel{\eqref{eq:int-app}}{\geq} \sum_{i \in [n]} (1 - p_i) \stackrel{\eqref{eq:opt2-continuous-app}}{\geq}  \Pr[\opt > x],$$
which concludes the proof of this case.

%% file: app-prelim.tex
\section{Proofs of Claims in the Preliminaries}\label{app:prelim}

\lemff*
\begin{proof}
We sample values $V_1, V_2$ and times $t_1, t_2$ for them. 
Let $i = \arg\max_{j \in [1,2]} V_j$.
\begin{eqnarray}
p_{\tau}(t,F_1,F_2) &=& 1-(1- p_{\tau}(t,F_1))\cdot(1- p_{\tau}(t,F_2)) \nonumber\\ &=& p_{\tau}(t,F_1) + p_{\tau}(t,F_2) - p_{\tau}(t,F_1) \cdot p_{\tau}(t,F_2) \nonumber\\&=& \Pr[t_1< t \wedge  V_1 > \tau(t_1) ] + \Pr[t_2< t \wedge  V_2 > \tau(t_2) ] \nonumber \\&-& \Pr[t_1< t \wedge  V_1 > \tau(t_1) ] \cdot \Pr[t_2< t \wedge  V_2 > \tau(t_2) ] \nonumber\\&\geq& \Pr[t_i< t \wedge  V_i > \tau(t_i) ] = p_{\tau}(t,F_1\cdot F_2). \nonumber
\end{eqnarray}
This concludes the proof of the lemma.
\end{proof}

\lemfff*
\begin{proof}
We prove the above inequality for every realization of the multiset $\{t_1,t_2\} = \{r,s\}$, where $r \leq  s$ (i.e., we sample the pair $(r,s)$ of arrival times of $V_1,V_2$, but not the corresponding matching between $V_1,V_2$ to $r,s$).   
Denote the event of the pair of times $t_1, t_2$ as a multiset being $\{r, s\}$ as $\mathcal{E}_{r,s}$.
We distinguish between the following cases:

\paragraph{Case 1: $r,s\geq t$.} The probability that the algorithm defined by $\tau$ stops before time $t$ is $0$ for both sides of the equation.

\paragraph{Case 2: $r<t \leq s$.} 
It holds that 
\begin{eqnarray}
0 \leq \frac{1}{2} \left(\sqrt{F_1(\tau(r))} - \sqrt{F_2(\tau(r))}\right)^2 = \frac{1}{2} F_1(\tau(r)) + \frac{1}{2} F_2(\tau(r)) - \sqrt{F_1 \cdot F_2(\tau(r))}. \label{eq:square1}
\end{eqnarray}
Thus we have
\begin{eqnarray}
p_{\tau}(t,F_1,F_2 \mid \mathcal{E}_{r,s})  &=& \frac{1}{2} (1 - F_1(\tau(r))) + \frac{1}{2} (1 - F_2(\tau(r))) \nonumber \\ &=& 1 - \frac{1}{2} F_1(\tau(r)) - \frac{1}{2} F_2(\tau(r)) \nonumber \\ &\stackrel{\eqref{eq:square1}}{\leq}& 1 - \sqrt{F_1 \cdot F_2(\tau(r))} \nonumber \\ &=& p_{\tau}(t,\sqrt{F_1\cdot F_2},\sqrt{F_1\cdot F_2} \mid \mathcal{E}_{r,s}). \nonumber
\end{eqnarray}

\paragraph{Case 3: $r,s<t$.} 
It holds that
\begin{eqnarray}
0 &\leq& \frac{1}{2} \left( \sqrt{F_1(\tau(r)) F_2(\tau(s))} - \sqrt{F_1(\tau(s)) F_2(\tau(r))} \right)^2 \nonumber \\ &=& \frac{1}{2} F_1(\tau(r)) F_2(\tau(s)) + \frac{1}{2} F_1(\tau(s)) F_2(\tau(r)) - \sqrt{F_1 \cdot F_2 (\tau(r))} \cdot \sqrt{F_1 \cdot F_2 (\tau(s))}. \label{eq:square2}
\end{eqnarray}
Thus we have
\begin{eqnarray}
1-p_{\tau}(t,F_1,F_2 \mid \mathcal{E}_{r,s})  &=& \frac{1}{2} F_1(\tau(r)) F_2(\tau(s)) + \frac{1}{2} F_1(\tau(s)) F_2(\tau(r)) \nonumber \\ &\stackrel{\eqref{eq:square2}}{\geq}& \sqrt{F_1 \cdot F_2 (\tau(r))} \cdot \sqrt{F_1 \cdot F_2 (\tau(s))} \nonumber \\ &=& 1-p_{\tau}(t,\sqrt{F_1\cdot F_2},\sqrt{F_1\cdot F_2} \mid \mathcal{E}_{r,s}). \nonumber
\end{eqnarray}
Therefore, Equation~\eqref{eq:f1f2} holds for all realizations of the multi-set $\{t_1,t_2\}$, and thus, it holds also in expectation. This concludes the proof of the lemma.
\end{proof}

%% file: App-new.tex
\section{Hardness for Activation-Based  Algorithms}\label{app:ex-activation}

In this section, we show a tight bound on the competition complexity of activation-based algorithms.
\begin{theorem}
The $(1-\epsilon)$-competition complexity of the class of  activation-based algorithms is $\Theta\left(\frac{\ln(\nicefrac{1}{\epsilon})}{\ln\ln(\nicefrac{1}{\epsilon})}\right)$. \label{thm:active-lower}
\end{theorem}
\begin{proof}
The proof of the upper bound follow from Lemma~\ref{lem:blind-upper} since a time-based threshold algorithm is a special case of activation-based algorithm.  We next show a lower bound by showing that $ k = \frac{\ln(\nicefrac{1}{\epsilon})}{4\ln\ln(\nicefrac{1}{\epsilon})}$ is not sufficient (we assume for simplicity of the proof that $k$ is a large enough integer). Consider an instance consisting of two types of rewards:
the first has a deterministic value of 1, the second is random and has value $1+\sqrt{\epsilon}$ with probability $1-p$ and value 0 with probability $p$, where $p = \frac{1}{\ln(\nicefrac{1}{\epsilon})}$.

The expected value of the prophet is $1+\sqrt{\epsilon}-p\cdot \sqrt{\epsilon}$.
Thus, giving $(1-\epsilon)$-approximation means that the algorithm needs to have an expected reward of at least 
$$ (1-\epsilon)\cdot \E[\opt] = (1-\epsilon)(1+\sqrt{\epsilon}-p\cdot \sqrt{\epsilon}) > 1+\sqrt{\epsilon}-3p\cdot \sqrt{\epsilon} .$$

An algorithm is defined by an activation function $g_{i,j}^v:[0,1]\rightarrow [0,1]$ for each reward $(i,j)$, and a value $v$ in the support of $V_{i,j}$.
We assume without loss of generality  that the activation of $0$ is always $0$ and the activation of $1+\sqrt{\epsilon}$ is always $1$ (since it is the maximum value in the support of the instance). This defines the functions of $g_{2,j}^{v}$ for $v\in \{0,1+\sqrt{\epsilon}\}$.
Thus, an algorithm is defined by the functions $g_{1,j}^{1}:[0,1]\rightarrow [0,1]$ for $j\in[k]$.
We distinguish between two cases:
\paragraph{Case 1: $\Pr\left[\exists j\in [k] \mbox{ such that } t_{1,j}\leq \frac{2}{k} \mbox{ and } V_{1,j} \mbox{ is activated} \right] \geq \frac{1}{k}$.}

In this case, we say that an order of arrival of $N$ rewards is ``bad'' if all the random rewards (rewards of the second type) arrive after time $\frac{2}{k}$.
The probability of a ``bad'' order is 
\begin{equation}    
\Pr[ \mbox{``bad'' order} ] = \left(1-\frac{2}{k}\right)^{k} \geq \frac{1}{10}, \label{eq:bad-non-new}
\end{equation}
where the last inequality holds for every $k\geq 8$.

By the definition of the case we are considering, it holds that 
\begin{equation}
    \Pr[\ALG =1 \mid \text{``bad'' order}] \geq \frac{1}{k}, \label{eq:1bad-new}
\end{equation}
since whether a reward of the deterministic type is activated is independent from the arrival times of the random rewards.

Thus, the performance of $\ALG$ is bounded by 
$$ \E[\ALG] \leq 1+\sqrt{\epsilon}-\sqrt{\epsilon}  \cdot\Pr[\ALG =1 \mid \text{``bad'' order} ] \cdot \Pr[ \mbox{``bad'' order} ]\stackrel{\eqref{eq:bad-non-new},\eqref{eq:1bad-new}}{\leq} 1+\sqrt{\epsilon}-\sqrt{\epsilon} \cdot \frac{1}{k} \cdot \frac{1}{10} = 1+\sqrt{\epsilon}-\frac{\sqrt{\epsilon}}{10k}.$$
For sufficiently small $\epsilon$ (such that $k \geq 61$), it holds that $3p< \frac{1}{10k}$, which means that $\ALG$ does not provide a $(1-\epsilon)$-approximation, which concludes the proof of this case.

\paragraph{Case 2: $\Pr\left[\exists j\in [k] \mbox{ such that } t_{1,j}\leq \frac{2}{k} \mbox{ and } V_{1,j} \mbox{ is activated} \right] < \frac{1}{k}$.}
For every $j \in [k]$ it holds that
\begin{eqnarray}
\Pr[V_{1,j} \mbox{ is activated}] &\leq& \Pr\left[t_{1,j} \geq \frac{2}{k}\right] + \Pr\left[t_{1,j} \leq \frac{2}{k} \mbox{ and } V_{1,j} \mbox{ is activated}\right] \nonumber \\ &\leq& 1-\frac{2}{k}+\Pr\left[\exists j' \in [k] : t_{1,j'} \leq \frac{2}{k} \mbox{ and } V_{1,j'} \mbox{ is activated}\right] \leq 1-\frac{1}{k} \label{eq:is_act}
\end{eqnarray}
Therefore, it holds that
\begin{eqnarray}
\Pr[\ALG =0] \geq p^k \cdot \prod_{j \in [k]} \Pr[V_{1,j} \mbox{ is \textbf{not} activated}] \stackrel{\eqref{eq:is_act}}{\geq} p^k \cdot \left(\frac{1}{k}\right)^k  \geq p^{2k} = \sqrt{\epsilon}\label{eq:alg_0}
\end{eqnarray}
Thus, the performance of $\ALG$ is bounded by 
\begin{eqnarray}
     \E[\ALG]  & \leq & (1+\sqrt{\epsilon})(1-\Pr[\ALG =0])  \stackrel{\eqref{eq:alg_0}}{\leq} (1+\sqrt{\epsilon})(1 -\sqrt{\epsilon})=  1-\epsilon< 1+\sqrt{\epsilon} -3p\cdot\sqrt{\epsilon} . \nonumber 
\end{eqnarray}
This means that $\ALG$ does not provide a $(1-\epsilon)$-approximation, which concludes the proof of the theorem.
\end{proof}